\documentclass[cmp]{svjour}
\usepackage{amssymb, amsmath,latexsym}

\journalname{Communications in Mathematical Physics}

\newtheorem{thm}    {Theorem}
\newtheorem{lem}     {Lemma}

\def\Tr{\mathop{\rm Tr}\nolimits}

\def\argmax{\mathop{\rm argmax}}

\def\Label#1{\label{#1}\ [\ #1\ ]\ }
\def\Label{\label}

\begin{document}

\title{{Universal coding for classical-quantum channel}}
\titlerunning{Universal coding for classical-quantum channel}

\author{Masahito Hayashi}
\institute{Graduate School of Information Sciences, Tohoku University, Sendai, 980-8579, Japan.\\ 
\email{hayashi@math.is.tohoku.ac.jp}}
\authorrunning{Masahito Hayashi}

\date{Received:}
\communicated{name}

\maketitle
\begin{abstract}
We construct a universal code for 
stationary and memoryless classical-quantum channel
as a quantum version of the universal coding by Csisz\'{a}r and K\"{o}rner.
Our code is constructed
by the combination of irreducible representation,
the decoder introduced through quantum information spectrum,
and the packing lemma.
\end{abstract}

\section{Introduction}\Label{s1}
The channel coding theorem for a stationary and 
memoryless\footnote{
Throughout the paper, a stationary memoryless channel
without using entangled input states is 
simply referred to as a stationary memoryless channel.}
(classical-)quantum channel has been established by combining the direct part shown by 
Holevo \cite{Holevo-QCTh} and 
Schumacher-Westmoreland 
\cite{Schumacher-Westmoreland} 
with the (weak) converse (impossible) part which goes back to 
1970's works by 
Holevo\cite{Holevo-bounds,Holevo-bounds2}. 
Its strong converse part has been shown by
Ogawa and Nagaoka\cite{Oga-Nag:channel} and Winter\cite{Winter}.
This theorem is a fundamental element of quantum information theory\cite{H-book}. 
After their achievement,
Ogawa and Nagaoka \cite{ON07} and Hayashi and Nagaoka\cite{Hay-Nag}
constructed other codes attaining the capacity. 
However, since the existing codes depend on the form of the channel, they are not robust against 
the disagreement between the sender's frame and receiver's frame.
In the classical system, 
Csisz\'{a}r and K\"{o}rner \cite{CK} constructed a universal channel coding, whose construction does not depend on the channel and depends only on the mutual information and the `type' of the input system, i.e., the empirical distribution of code words, whose precise explanation will be explained in Section \ref{s3}.
Such a universal code for the quantum case was also constructed for 
variable-length source coding\cite{Ly,Da} and fixed-length source coding\cite{CK}.

Concerning the quantum system, Jozsa et al. \cite{JH} constructed 
a universal fixed-length source coding, which depends only on the compression rate and attains the minimum compression rate.
Hayashi \cite{Expo-s} discussed the exponential decreasing rate of its decoding error.
Further, Hayashi and Matsumoto \cite{HayaMa} constructed a universal variable-length source coding in the quantum system.
However, any universal coding for classical-quantum channel was not constructed.
In fact,
the universal coding is required
when the receiver cannot synchronize his frame with the sender's frame.

In the present paper, we construct a universal coding for a classical-quantum channel, which attains the quantum mutual information and depends only on the coding rate and the `type' of the input system.
In the proposed construction, the following three methods play essential roles.
One is the decoder given by the proof of the information spectrum method.
In the information spectrum method, the decoder is constructed by the square root measurement of the projectors given by the quantum analogue of the likelihood ratio between the signal state and the mixture state\cite{Hay-Nag,Verdu-Han}.

The second method is the irreducible decomposition of the dual representation of the special unitary group and the permutation group.
The method of irreducible decomposition provides the universal protocols in quantum setting\cite{JH,Haya97,H2001,MaHa07,HayaMa,KeylW,Sanov}. 
However, even in the classical case, 
the universal channel coding requires the conditional type as well as the type\cite{CK}.
In the present paper, we introduce a quantum analogue of the conditional type, which is the most essential part of the present paper.

The third method is the packing lemma, which yields a suitable combination of the signal states independent of the form of the channel in the classical case\cite{CK}.
This method plays the same role in the present paper.

The remainder of the present paper is organized as follows.
In section \ref{s2}, we give the notation herein and the main result including the existence of a
universal coding for classical-quantum channel.
In this section,
we presented the exponential decreasing rate of the error probability of the presented universal code.
In section \ref{s3}, the notation for group representation theory is presented and
a quantum analogue of conditional type is introduced.
In section \ref{s4}, we give a code that well works universally.
In section \ref{s5}, 
the exponential decreasing rate mentioned in section \ref{s2} is proven by using
the property given in section \ref{s3}.

\section{Main Result}\Label{s2}
In the classical-quantum channel, we focus on the set of input alphabets 
${\cal X}:=\{1, \ldots, k\}$ and the representation space ${\cal H}$ of the output system, whose dimension is $d$.
Then, a classical-quantum channel is given as the map from ${\cal X}$ to the set of densities on ${\cal H}$ with the form $i \mapsto W(i)$.
The $n$-th discrete memoryless extension is given as
the map from ${\cal X}^n$ to the set of densities on the $n$-th tensor product system ${\cal H}^{\otimes n}$. That is, this extension maps the input sequence $\vec{i}=(i_1, \ldots, i_n)$ to the state $W_n(\vec{i}_n):=W(i_1)\otimes \cdots \otimes W(i_n)$.
Sending the message $\{1, \ldots, M_n\}$ requires an encoder and a decoder.
The encoder is given as a map $\varphi_n$ from the set of messages $\{1, \ldots, M_n\}$ to the set of alphabets ${\cal X}^n$, and the decoder is given by a POVM $Y^n=\{Y_i^n\}_{i=1}^{M_n}$.
Thus, the triplet $\Phi_n:=(M_n, \varphi_n, Y)$ is called a code.
Its performance is evaluated by the size $|\Phi_n|:=M_n$ and
the average error probability given by 
\begin{align*}
\varepsilon [\Phi_n,W]:=\frac{1}{M_n} \sum_{i=1}^{M_n} \Tr W_n (\varphi_n(i)) (I-Y_i^n).
\end{align*}
As mentioned in the following main theorem, there exists an asymptotically optimal code that depends only on the coding rate.

\begin{thm}\Label{thm1}
For any distribution $\vec{p}=\{p_i\}_{i=1}^k$ on the set of input alphabets ${\cal X}:=\{1, \ldots, k\}$
and any real number $R$, 
there is a sequence of codes $\{\Phi_n\}_{n=1}^{\infty}$ such that
\begin{align*}
\lim_{n \to \infty}\frac{-1}{n}\log \varepsilon [\Phi_n,W]
& \ge \max_{0 \le t \le 1} \frac{\phi_{W,\vec{p}}(t)-tR}{1+t} \\
\lim_{n \to \infty}\frac{1}{n}\log |\Phi_n| &= R
\end{align*}
for any classical-quantum channel $W$,
where $\phi_{W,\vec{p}}(t)$ is given by
\begin{align*}
\phi_{W,\vec{p}}(t):=- (1-t) \log \Tr (\sum_{i=1}^k p_i W(i)^{1-t})^{\frac{1}{1-t}} .
\end{align*}
Note that the code $\{\Phi_n\}_{n=1}^{\infty}$ does not depend on the channel $W$, and depends only on the 
distribution $\vec{p}$ and the coding rate $R$.
\end{thm}
The derivative of $\phi_{W,\vec{p}}(t)$ is given as
\begin{align*}
\phi_{W,\vec{p}}'(0)&=I(p,W):=\sum_{i=1}^k p_i \Tr W(i)(\log W_i- \log W_{\vec{p}}) \\
W_{\vec{p}}&:=\sum_{i=1}^k p_i W(i).
\end{align*}
When the transmission rate $R$ is smaller than the mutual information $I(\vec{p},W)$,
\begin{align*}
\max_{0 \le t \le 1} \frac{\phi_{W,\vec{p}}(t)-tR}{1+t} >0
\end{align*}
because there exists a parameter $t \in (0,1)$ such that $\phi_{W,\vec{p}}(t) -tR >0$.
That is, the average error probability $\varepsilon [\Phi_n,W]$ goes to zero.

\section{Group representation theory}\Label{s3}
In this section, we focus on the dual representation on the $n$-fold tensor product space by the 
the special unitary group $SU(d)$ and the $n$-th symmetric group $S_n$\footnote{Christandl\cite{Christandl} contains
a good survey of representation theory for quantum information.}.
For this purpose, we focus on the Young diagram and the `type'. 
The former is a key concept in group representation theory and the latter is that in information theory\cite{CK}.
When the vector of integers $\vec{n}=(n_1,n_2, \ldots, n_d)$ satisfies the condition $n_1 \ge n_2 \ge \ldots \ge n_d\ge 0$ and $\sum_{i=1}^d n_i=n$,
the vector $\vec{n}$ is called the Young diagram (frame) with size $n$ and depth $d$, the set of which is denoted as $Y_n^d$.
When the vector of integers $\vec{n} $ satisfies the condition $n_i \ge 0$ and $\sum_{i=1}^d n_i=n$,
the vector $\vec{p}=\frac{\vec{n}}{n}$ is called the `type' with size $n$, the set of which is denoted as $T_n^d$.
Further, 
for $\vec{p} \in T_n^d$, the subset of ${\cal X}^n$ is defined as:
\begin{align*}
T_{\vec{p}}:= \{\vec{x} \in {\cal X}^n| \hbox{The empirical distribution of } \vec{x}
\hbox{ is equal to }\vec{p}\}.
\end{align*}
The numbers of these sets are evaluated as follows:
\begin{align}
|Y_n^d| &\le |T_n^d|\le (n+1)^{d-1}\Label{9}\\
(n+1)^{-d} e^{n H(\vec{p})} &\le |T_{\vec{p}}| ,\Label{7}
\end{align}
where $H(\vec{p}):= -\sum_{i=1}^d p_i \log p_i$\cite{CK}.
Using the Young diagram, the irreducible decomposition of the above representation can be characterized as follows:
\begin{align*}
{\cal H}^{\otimes n}=\bigoplus_{\vec{n}\in Y_n^d} {\cal U}_{\vec{n}} \otimes {\cal V}_{\vec{n}},
\end{align*}
where ${\cal U}_{\vec{n}}$ is the irreducible representation space of $SU(d)$ characterized by $\vec{n}$, and
${\cal V}_{\vec{n}}$ is the irreducible representation space of $n$-th symmetric group $S_n$ characterized by $\vec{n}$.
Here, the representation of the $n$-th symmetric group $S_n$ is denoted as $V: s \in S_n \mapsto V_s$.
For $\vec{n}\in Y_n^d$, the dimension of ${\cal U}_{\vec{n}}$ is evaluated by
\begin{align}
\dim {\cal U}_{\vec{n}} \le n^{\frac{d(d-1)}{2}} \Label{11}.
\end{align}
Then, denoting the projection to the subspace ${\cal U}_{\vec{n}} \otimes {\cal V}_{\vec{n}}$ as $I_{\vec{n}}$, we define the following.
\begin{align}
\rho_{\vec{n}}&:= \frac{1}{\dim {\cal U}_{\vec{n}} \otimes {\cal V}_{\vec{n}}} I_{\vec{n}}\\
\rho_{U,n}&:= \sum_{\vec{n}\in Y_n^d} \frac{1}{|Y_n^d|} \rho_{\vec{n}}.\Label{15}
\end{align}
Any state $\rho$ and any Young diagram $\vec{n}\in Y_n^d$ satisfy the following:
\begin{align*}
\dim {\cal U}_{\vec{n}} \rho_{\vec{n}} 
\ge I_{\vec{n}} \rho^{\otimes n} I_{\vec{n}}.
\end{align*}
Thus, (\ref{9}), (\ref{11}), and (\ref{15}) yield 
the inequality
\begin{align}
n^{\frac{d(d-1)}{2}} |Y_n^d| \rho_{U,n} \ge \rho^{\otimes n}.\Label{ineq-2}
\end{align}

Next, we focus on two systems ${\cal X}$ and ${\cal Y}=\{1,\ldots, l\}$.
When the distribution of ${\cal X}$ is given by a probability distribution $\vec{p}=(p_1, \ldots, p_d)$ on $\{1, \ldots, d\}$, 
and the conditional distribution on ${\cal Y}$ with the condition on ${\cal X}$ is given by $\vec{V}$,
we denote the joint distribution 
on ${\cal X} \times {\cal Y}$ by $\vec{p} \vec{V}$ and
the distribution on ${\cal Y}$ by $\vec{p} \cdot \vec{V}$.
When the empirical distribution of $\vec{x}\in {\cal X}^n$
is $(\frac{n_1}{n},\ldots, \frac{n_d}{n})$,
the sequence of types $\vec{V}=(\vec{v}_1, \ldots, \vec{v}_d)\in 
T_{n_1}^{l}\times \cdots \times T_{n_d}^{l}$ is called a conditional type for $\vec{x}$\cite{CK}.
We denote the set of conditional types for $\vec{x}$ by $V(\vec{x},{\cal Y})$.
For any conditional type $V$ for $\vec{x}$, we define the subset of ${\cal Y}^n$:
\begin{align*}
T_{\vec{V}}(\vec{x}) 
:= \left\{\vec{y} \in {\cal Y}^n\left|
\begin{array}{l}
\hbox{The empirical distribution of } \\
((x_1,y_1), \ldots, (x_n,y_n))
\hbox{ is equal to }
\vec{p} \vec{V}.
\end{array}
\right.\right\},
\end{align*}
where 
$\vec{p}$ is the empirical distribution of $\vec{x}$.

We define the state $\rho_{\vec{x}}$ for $\vec{x} \in {\cal X}^n$.
For this purpose, we consider a special element $\vec{x}'=(\underbrace{1,\ldots,1}_{m_1},\underbrace{2,\ldots,2}_{m_2},\ldots, 
\underbrace{k,\ldots, k}_{m_k})$.
The state $\rho_{\vec{x}'}$ is defined as
$\rho_{\vec{x}'}:= \rho_{U,m_1} \otimes \rho_{U,m_2} \otimes \cdots \otimes \rho_{U,m_k}$.
For a general element $\vec{x} \in {\cal X}^n$,
we choose a permutation $s\in S_n$ such that $\vec{x}=s\vec{x}'$.
Then, we define the state $\rho_{\vec{x}}$ is defined as
$\rho_{\vec{x}}:=U_s \rho_{\vec{x}'}U_s^\dagger$, where $U_s$ is the unitary representation of $S_n$.
This state plays a similar role as the conditional type in the classical case.
Using the inequality (\ref{ineq-2}),
we have
\begin{align}
n^{\frac{kd(d-1)}{2}} |Y_n^d|^k
\rho_{\vec{x}} \ge W_n(\vec{x}).
\Label{ineq-1}
\end{align}
For $\vec{n}_1\in Y_{m_1}^d, \vec{n}_2\in Y_{m_2}^d, \ldots, \vec{n}_k\in Y_{m_k}^d$,
the density
$\rho_{\vec{n}_1}\otimes \rho_{\vec{n}_2}\otimes \cdots \otimes
\rho_{\vec{n}_k}$
is commutative with the projector $I_{\vec{n}}$ for $\vec{n} \in Y_n^d$.
This fact implies that
the density $\rho_{\vec{x}}$ is commutative with the density $\rho_{U,n}$.
This property is essential for the construction of the proposed decoder.

\section{Construction of code}\Label{s4}
According to Csisz\'{a}r and K\"{o}rner\cite{CK},
the proposed code is constructed as follows.
\begin{lem}
For a positive number $\delta>0$,
a type $\vec{p} \in T_n^d$,
and a real positive number $R< H(\vec{p})$,
there exist
$M_n:=e^{n(R-\delta)}$
distinct elements 
${\cal M}_n:=\{\vec{x}_1,\ldots, \vec{x}_{M_n}\} \subset T_{\vec{p}}$ 
such that
their empirical distributions are $\vec{p}$ and
\begin{align*}
|T_{\vec{V}}(\vec{x}) \cap ({\cal M}_n\setminus \{\vec{x}\})|
\le |T_{\vec{V}}(\vec{x})| 
e^{-n(H(\vec{p})-R)} 
\end{align*}
for $\vec{x} \in {\cal M}_n \subset T_{\vec{p}}$ and $\vec{V} \in V(\vec{x},{\cal X})$.
\end{lem}
This lemma can be shown by substituting the identical map into $\hat{V}$ in Lemma 5.1 in 
Csisz\'{a}r and K\"{o}rner\cite{CK}.
Since Csisz\'{a}r and K\"{o}rner proved Lemma 5.1 using the random coding method,
we can replace $\delta$ by $\frac{1}{\sqrt{n}}$.
That is, 
there exist
$M_n:= e^{n R- \sqrt{n}}$
distinct elements 
${\cal M}_n:=\{\vec{x}_1,\ldots, \vec{x}_{M_n}\} \subset T_{\vec{p}}$ 
such that
their empirical distributions are $\vec{p}$ and
\begin{align}
|T_{\vec{V}}(\vec{x}) \cap ({\cal M}_n\setminus \{\vec{x}\})|
\le |T_{\vec{V}}(\vec{x})| 
e^{-n(H(\vec{p})-R)} 
 \Label{20}
\end{align}
for $\vec{x} \in {\cal M}_n \subset T_{\vec{p}}$ and $\vec{V} \in V(\vec{x},{\cal X})$.
Now, we transform the property (\ref{20}) to a more useful form.

Using the encoder ${\cal M}_n$, we can define the distribution $P_{{\cal M}_n}$ as
\begin{align*}
p_{{\cal M}_n}(\vec{x})= 
\left\{
\begin{array}{cc}
\frac{1}{|{\cal M}_n|} & \vec{x} \in {\cal M}_n \\
0 & \vec{x} \notin {\cal M}_n.
\end{array}
\right.
\end{align*}
For any $\vec{x} \in {\cal X}^n$, we define the invariant subgroup 
$S_{\vec{x}}\subset S_n$:
\begin{align*}
S_{\vec{x}}: = \{s \in S_n | s(\vec{x})=\vec{x} \}.
\end{align*}
Since $\vec{x}' \in T_{\vec{p}}$ implies that
\begin{align*}
\vec{p}^{n}(\vec{x}')=e^{-nH(\vec{p})} ,
\end{align*}
any element $\vec{x}' \in T_{\vec{V}}(\vec{x}) \cap {\cal M}_n \subset T_{\vec{p}}$
satisfies
\begin{align}
&\sum_{s \in S_{\vec{x}}}
\frac{1}{|S_{\vec{x}}|}
p_{{\cal M}_n}\circ s (\vec{x}')
=
\frac{|T_{\vec{V}}(\vec{x}) \cap {\cal M}_n|}{|T_{\vec{V}}(\vec{x})|}
\cdot
\frac{1}{|{\cal M}_n|} 
=
\frac{|T_{\vec{V}}(\vec{x}) \cap ({\cal M}_n\setminus \{\vec{x}\})|}{|T_{\vec{V}}(\vec{x})||{\cal M}_n|} \nonumber \\
\le &
e^{-nH(\vec{p})}
e^{\sqrt{n} } 
=\vec{p}^{n}(\vec{x}') e^{\sqrt{n} } \Label{8}
\end{align}
when the conditional type $\vec{V}$ is not identical.
Relation (\ref{8}) holds for any $\vec{x}'(\neq \vec{x}) \in {\cal M}_n $
because there exists a conditional type $\vec{V}$ such that $\vec{x}' \in T_{\vec{V}}(\vec{x})$ and $\vec{V}$ is not identical.

Next, for any $\vec{x} \in {\cal X}^n$ and any real number $C_n$, 
we define the projection
\begin{align*}
P(\vec{x}):=
\{\rho_{\vec{x}}- C_n \rho_{U,n}\ge 0\},
\end{align*}
where 
$\{X\ge 0\}$ presents the projection $\sum_{i:x_i \ge 0}E_i$
for a Hermitian matrix $X$ with the diagonalization $X=\sum_ix_i E_i$.
Remember that the density $\rho_{\vec{x}}$ is commutative with the other density $\rho_{U,n}$.
Using the projection $P(\vec{x})$, we define the decoder:
\begin{align*}
Y_{\vec{x}'}:=
\sqrt{\sum_{\vec{x} \in {\cal M}_n}P(\vec{x})}^{-1}
P(\vec{x}')
\sqrt{\sum_{\vec{x} \in {\cal M}_n}P(\vec{x})}^{-1}.
\end{align*}
In the following, 
the above-constructed code $(e^{nR-\sqrt{n}}, {\cal M}_n,\{Y_{\vec{x}}\}_{\vec{x}\in {\cal M}_n})$ is denoted by $\Phi_{U,n}(\vec{p},R)$.

\section{Exponential evaluation}\Label{s5}
Hayashi and Nagaoka\cite{Hay-Nag} showed that
\begin{align*}
I- Y_{\vec{x}'} \le 2 (I- P(\vec{x}'))+ 4 \sum_{\vec{x} (\neq \vec{x}')\in {\cal M}_n}P(\vec{x}).
\end{align*}
Then, the average error probability of $\Phi_{U,n}(\vec{p},R)$ is evaluated by
\begin{align}
& \frac{1}{|{\cal M}_n|}
\sum_{\vec{x}' \in {\cal M}_n}
\Tr W_n(\vec{x}')(I- Y_{\vec{x}'}) \nonumber \\
\le & 
\frac{2}{|{\cal M}_n|}
\sum_{\vec{x}' \in {\cal M}_n} \Tr W_n(\vec{x}') (I- P(\vec{x}'))
+
\frac{4}{|{\cal M}_n|}
\sum_{\vec{x}' \in {\cal M}_n} 
\Tr W_n(\vec{x}') 
\sum_{\vec{x} (\neq \vec{x}')\in {\cal M}_n}P(\vec{x})\nonumber \\
=&
\frac{2}{|{\cal M}_n|}
\sum_{\vec{x} \in {\cal M}_n} \Tr W_n(\vec{x}) (I- P(\vec{x}))
\nonumber \\
&+
4 
\Tr 
\left[
\sum_{\vec{x} \in {\cal M}_n}P(\vec{x})
\left(\frac{1}{|{\cal M}_n|}
\sum_{\vec{x}' (\neq \vec{x}) \in {\cal M}_n} W_n(\vec{x}')\right) 
\right]\Label{37}.
\end{align}

Since the density $\rho_{\vec{x}}$ is commutative with the density $\rho_{U,n}$,
we have
\begin{align}
(I- P(\vec{x}))
=
\{\rho_{\vec{x}}- C_n \rho_{U,n}< 0\}
\le 
\rho_{\vec{x}}^{-t} C_n ^t \rho_{U,n} ^t \Label{ineq-4}
\end{align}
for $0 \le t \le 1$.
Since 
the density $\rho_{\vec{x}}$ is commutative with the density $W_n(\vec{x})$,
$W_n(\vec{x})\rho_{\vec{x}}^{-t} $ is a Hermite matrix and 
(\ref{ineq-1}) implies that
\begin{align}
W_n(\vec{x})\rho_{\vec{x}}^{-t} 
\le 
n^{\frac{ktd(d-1)}{2}} |Y_n^d|^{kt}
W_n(\vec{x})^{1-t}.\Label{ineq-5}
\end{align}
Using (\ref{ineq-4}) and (\ref{ineq-5}), we have
\begin{align}
& \Tr W_n(\vec{x}) (I- P(\vec{x}))
\le
\Tr 
W_n(\vec{x})\rho_{\vec{x}}^{-t} 
 \rho_{U,n} ^t C_n ^t \nonumber \\
\le &
n^{\frac{ktd(d-1)}{2}} |Y_n^d|^{kt} C_n ^t 
\Tr W_n(\vec{x})^{1-t}
\rho_{U,n} ^t .\Label{l11}
\end{align}
Since the quantity $\Tr W_n(\vec{x}) (I- P(\vec{x}))$ is invariant for the action of the permutation
and the relation (\ref{7}) implies that
\begin{align}
\vec{p}^{n}(\vec{x})= e^{-n H(\vec{p})} 
\ge \frac{(n+1)^{-d}}{|T_{\vec{p}}|} \Label{l12}
\end{align}
for $\vec{x} \in T_{\vec{p}}$,
we obtain
\begin{align}
& \Tr W_n(\vec{x}) (I- P(\vec{x}))
=
\frac{1}{|T_{\vec{p}}|}
\sum_{\vec{x}'\in T_{\vec{p}}}
\Tr W_n(\vec{x}') (I- P(\vec{x}')) \nonumber \\
\le &
(n+1)^d \sum_{\vec{x}'\in {\cal X}^n
}
\vec{p}^{n}(\vec{x}')
\Tr W_n(\vec{x}') (I- P(\vec{x}'))\Label{l5} \\
\le &
(n+1)^d n^{\frac{ktd(d-1)}{2}} |Y_n^d|^{kt} C_n ^t 
\Tr 
(\sum_{\vec{x}'\in {\cal X}^n}
\vec{p}^{n}(\vec{x}')
 W_n(\vec{x}')^{1-t})
\rho_{U,n} ^t \Label{l6} \\
\le &
(n+1)^{d+\frac{ktd(d-1)}{2}} |Y_n^d|^{kt} C_n ^t 
\max_{\sigma}
\Tr 
\left[\sum_{x\in {\cal X}} \vec{p}(x) W_n(x)^{1-t}\right]^{\otimes n}
\sigma^t \nonumber \\
\le &
(n+1)^{d+\frac{ktd(d-1)}{2}} |Y_n^d|^{kt} C_n ^t 
\left( \Tr 
\left( \left[\sum_{x\in {\cal X}} \vec{p}(x) W_n(x)^{1-t}\right]^{\otimes n}\right)
^{\frac{1}{1-t}}\right)^{1-t} \Label{l1}\\
= &
(n+1)^{d+\frac{ktd(d-1)}{2}} |Y_n^d|^{kt} C_n ^t 
\left(\Tr 
\left(\sum_{x\in {\cal X}} \vec{p}(x) W_n(x)^{1-t}\right)^{\frac{1}{1-t}}
\right)^{n(1-t)}  \nonumber \\
=& (n+1)^{d+\frac{ktd(d-1)}{2}} |Y_n^d|^{kt} C_n ^t 
e^{- n \phi_{W,\vec{p}}(t)},\Label{49}
\end{align}
where 
(\ref{l5}), (\ref{l6}), and (\ref{l1}) follow from 
(\ref{l12}), (\ref{l11}), and Lemma \ref{lem1} in Appendix, respectively.

Next, we evaluate the second term of (\ref{37}) using the invariant property of $S_{\vec{x}}$:
\begin{align}
& \Tr 
\left[
P(\vec{x})
\left(\frac{1}{|{\cal M}_n|}
\sum_{\vec{x}' (\neq \vec{x}) \in {\cal M}_n} W_n(\vec{x}')\right) 
\right]\nonumber \\
=&
\Tr 
\left[
P(\vec{x})
\sum_{\vec{x}' (\neq \vec{x}) \in {\cal M}_n} 
p_{{\cal M}_n}(\vec{x}')
W_n(\vec{x}')
\right] \nonumber\\
= &
\Tr 
\left[
P(\vec{x})
\sum_{s \in S_{\vec{x}}}
\frac{1}{|S_{\vec{x}}|}
\sum_{\vec{x}' (\neq \vec{x}) \in {\cal M}_n} 
p_{{\cal M}_n}(\vec{x}')
V_{s}
W_n(\vec{x}')V_{s}^{*}
\right] \nonumber\\
= &
\Tr 
\left[
P(\vec{x})
\sum_{\vec{x}' (\neq \vec{x}) \in {\cal M}_n} 
\sum_{s \in S_{\vec{x}}}
\frac{1}{|S_{\vec{x}}|}
p_{{\cal M}_n}\circ s^{-1}(\vec{x}')
W_n(\vec{x}')
\right] \nonumber\\
\le &
\Tr 
\left[
P(\vec{x})
\sum_{\vec{x}' (\neq \vec{x}) \in {\cal M}_n}
\vec{p}^{n}(\vec{x}') e^{\sqrt{n}} 
W_n(\vec{x}')
\right] \Label{l13} \\
= &
e^{\sqrt{n}} 
\Tr 
\left[
P(\vec{x})
W_{\vec{p}}^{\otimes n}
\right] \nonumber\\
\le &
e^{\sqrt{n}} 
\Tr 
\left[
P(\vec{x})
n^{\frac{d(d-1)}{2}} |Y^d_n|
\rho_{U,n}
\right] \Label{l14} \\
\le &
e^{\sqrt{n} } 
\Tr 
\left[
P(\vec{x})
n^{\frac{d(d-1)}{2}} |Y^d_n|
C_n^{-1} \rho_{\vec{x}}
\right] \Label{l15}  \\
\le &
e^{\sqrt{n} } 
\Tr 
\left[
n^{\frac{d(d-1)}{2}}|Y^d_n|
C_n^{-1} \rho_{\vec{x}}
\right] 
= 
e^{\sqrt{n} } 
n^{\frac{d(d-1)}{2}}|Y^d_n|
C_n^{-1},\Label{59}
\end{align}
where
(\ref{l13}), (\ref{l14}), and (\ref{l15})
follow from (\ref{8}), (\ref{ineq-2}), and 
the inequality
$P(\vec{x})(\rho_{U,n}- C_n^{-1} \rho_{\vec{x}}) \le 0$.

For any $t \in (0,1)$ and $R>0$, 
we choose 
$|{\cal M}_n|:=e^{nR-\sqrt{n}}$,
$C_n:=e^{n(R+r(t))}$, and
$r(t):=\frac{\phi_{W,\vec{p}}(t)-tR }{1+t}$.
Since $r(t)=\phi_{W,\vec{p}}(t)-t(R+r(t))$,
from (\ref{37}), (\ref{49}) and (\ref{59}), the exponential decreasing rate of the average error probability is evaluated as
\begin{align*}
\lim_{n\to \infty}
\frac{-1}{n}\log \varepsilon (\Phi_{U,n}(\vec{p},R),W)
\ge \min\{\phi_{W,\vec{p}}(t)-t(R+r(t)),r(t)\} =
\frac{\phi_{W,\vec{p}}(t)-tR }{1+t}.
\end{align*}
That is, when we choose $t_0:=\argmax_{t \in (0,1)}\frac{\phi_{W,\vec{p}}(t)-tR }{1+t}$,
$|{\cal M}_n|:=e^{nR-\sqrt{n}}$, and 
$C_n:=e^{n(R+r(t_0))}$,
we obtain
\begin{align*}
\lim_{n\to \infty}
\frac{-1}{n}\log \varepsilon (\Phi_{U,n}(\vec{p},R),W)
\ge \max_{t \in (0,1)}
\frac{\phi_{W,\vec{p}}(t)-tR }{1+t}
\end{align*}
for any channel $W$.
Therefore, we obtain Theorem \ref{thm1}.

\section{Discussion}\Label{s6}
We have constructed a universal code attaining the quantum mutual information
based on the combination of information spectrum method, group representation theory,
and the packing lemma.
The presented code well works because any tensor product state $\rho^{\otimes n}$
is close to the state $\rho_{U,n}$.
Indeed, 
Krattenthaler and Slater \cite{KS} demonstrated the existence of 
the state $\sigma_n$ such that
$\frac{1}{n}D(\rho^{\otimes n}\|\sigma_n)\to n$ 
for any state $\rho$ in the qubit system
as a quantum analogue of Clarke and Barron's result\cite{CB}.
Its $d$-dimensional extension is discussed in another paper\cite{prep}. 

Further, Hayashi \cite{Expo-c} derived an exponential decreasing rate of error probability in classical-quantum channel, which is $
\max_{t:0\le t\le 1}
-(\log \sum_i p_i \Tr [W(i)^{1-t} W_p^t]) -tR$.
Since 
\begin{align*}
& e^{-\frac{\phi_{W,\vec{p}}(t)-t(R+r(t))}{1+t}}
=
e^{-(\phi_{W,\vec{p}}(t)-t(R+r(t)))}
=
e^{t(R+r(t))}
\max_\sigma 
\Tr (\sum_i p_i W(i)^{1-t})\sigma^t \\
\ge &
e^{tR}
\Tr (\sum_i p_i W(i)^{1-t})(\sum_i p_i W(i))^t
=
e^{-
(-(\log \sum_i p_i \Tr [W(i)^{1-t} W_{\vec{p}}^t]) -tR)},
\end{align*}
we obtain 
\begin{align*}
\max_{t:0\le t\le 1}
-(\log \sum_i p_i \Tr [W(i)^{1-t} W_{\vec{p}}^t]) -tR
\ge 
\max_{t:0\le t\le 1}
\frac{\phi_{W,\vec{p}}(t)-tR}{1+t}.
\end{align*}
That is, the obtained exponential decreasing rate is smaller than that of Hayashi\cite{Expo-c}.
However, according to Csisz\'{a}r and K\"{o}rner \cite{CK},
the exponential decreasing rate of the universal coding is the same as the optimal exponential decreasing rate in the classical case when the rate is close to the capacity.
Hence, if a more sophisticated evaluation is applied,
a better exponential decreasing rate can be expected.
Such an evaluation is left as a future problem.

\section*{Acknowledgment}
This research
was partially supported by a Grant-in-Aid for Scientific Research on Priority Area `Deepening and Expansion of Statistical Mechanical Informatics (DEX-SMI)', No. 18079014 and
a MEXT Grant-in-Aid for Young Scientists (A) No. 20686026.

\appendix
\section{Maximization}
The following lemma is used for the derivation in Section \ref{s5}.
\begin{lem}\Label{lem1}
When $X$ is a positive semi-definite,
we have
\begin{align}
\max_{\sigma} \Tr X \sigma^t= (\Tr X^{\frac{1}{1-t}})^{1-t}\Label{l3}
\end{align}
for $0 \le t \le 1$,
where $\sigma$ is a density matrix.
\end{lem}

\begin{proof}
First, we prove 
\begin{align}
\max_{q_i\ge 0:\sum_i q_i=1}
\Tr X \sum_{i}q_i^t |i\rangle \langle i|=
\left(\sum_{i}
\langle i|X |i\rangle^{\frac{1}{1-t}}
\right)^{1-t}\Label{l4}
\end{align}
by the Lagrange multiplier method.
Let $\lambda$ be the Lagrange multiplier.
Then, 
\begin{align*}
0=
\sum_i (\langle i|X |i\rangle t q_i^{t-1}+ \lambda )\delta q_i
\end{align*}
Thus, 
\begin{align*}
0=\langle i|X |i\rangle t q_i^{t-1}+ \lambda.
\end{align*}
That is,
\begin{align*}
-\frac{t}{\lambda}\langle i|X |i\rangle  = q_i^{1-t} .
\end{align*}
Then, when the maximizing $q_i$ has the form
$C \langle i|X |i\rangle^{\frac{1}{1-t}}$ with the normalizing constant $C$,
the constant $C$ has the form $C=\frac{1}{\sum_j \langle j|X |j\rangle^{\frac{1}{1-t}}}$.
Substituting 
$\frac{\langle i|X |i\rangle^{\frac{1}{1-t}}}{\sum_j \langle j|X |j\rangle^{\frac{1}{1-t}}}$
into
$q_i$,
we obtain (\ref{l4}).

Since 
\begin{align*}
\left(\sum_{i}
\langle i|X |i\rangle^{\frac{1}{1-t}}
\right)^{1-t}
=
\Tr X \left(\sum_{i}
\langle i|\frac{1}{\Tr X}X |i\rangle^{\frac{1}{1-t}}
\right)^{1-t},
\end{align*}
the maximum
$\max_{\sigma} \Tr X \sigma^t$ is given when we choose the basis $\{|i\rangle \}$
maximizing $\sum_{i}
\langle i|\frac{1}{\Tr X}X |i\rangle^{\frac{1}{1-t}}$.
Since the function $x \mapsto x^{\frac{1}{1-t}}$ is a convex function,
$
\langle i|\frac{1}{\Tr X}X |i\rangle^{\frac{1}{1-t}}
\le
\langle i|(\frac{1}{\Tr X}X)^{\frac{1}{1-t}}
 |i\rangle$.
Therefore,
\begin{align*}
\left(\sum_{i}
\langle i|X |i\rangle^{\frac{1}{1-t}}
\right)^{1-t}
\le
(\Tr X^{\frac{1}{1-t}})^{1-t}.
\end{align*}
The equality holds when we choose the basis $\{|i\rangle\}$
as the eigenvectors of $X$.
Therefore, we obtain (\ref{l3}).
\end{proof}


\begin{thebibliography}{99}
\bibitem{Holevo-QCTh}
A.S.\ Holevo,
``The capacity of the quantum channel with general signal states,''
{\em IEEE Trans. Inform. Theory}, vol.44, 269--273, 1998.
\bibitem{Schumacher-Westmoreland}
B.\ Schumacher and M.D.\ Westmoreland,
``Sending classical information via noisy quantum channels,''
{\em Phys. Rev. A}, vol.56, 131--138, 1997.
\bibitem{Holevo-bounds}
A.S.\ Holevo,
``Bounds for the quantity of information transmitted by a quantum
  communication channel,''
{\em Probl. Inform. Transm.}, vol.9, 177--183, 1973.
\bibitem{Holevo-bounds2}
A.S.\ Holevo,
``On the capacity of quantum communication channel,''
{\em Probl. Inform. Transm.}, vol. 15, no. 4, pp. 247--253, 1979.
\bibitem{Oga-Nag:channel}
T. Ogawa and H. Nagaoka, 
``Strong Converse to the Quantum Channel Coding Theorem," 
{\em IEEE Trans. Inform. Theory}, 
vol.45, 2486-2489, 1999.
\bibitem{Winter}
A.\ Winter,
``Coding theorem and strong converse for quantum channels,''
{\em IEEE Trans. Inform. Theory}, 
vol.45, 2481-2485, 1999.
\bibitem{H-book}
M. Hayashi, {\em Quantum Information: An Introduction}, (Springer, Berlin, 2006).

\bibitem{ON07}
T. Ogawa and H. Nagaoka,
``Making Good Codes for Classical-Quantum Channel Coding via Quantum Hypothesis Testing,''
{\em IEEE Trans. Inform. Theory}, vol.53, 2261 - 2266, (2007).

\bibitem{Hay-Nag}
M. Hayashi and H. Nagaoka:
``General formulas for capacity of classical-quantum channels,''
{\em IEEE Trans. Infor. Theory},
{\bf 49}, 1753-1768 (2003).

\bibitem{CK}
I.\ Csisz\'{a}r and J.\ K\"{o}rner,
{\it Information Theory:\ Coding Theorems for Discrete Memoryless Systems},
(Academic Press, 1981).

\bibitem{Ly} T. J. Lynch,
``Sequence time coding for data compression,''
{\em Proc. IEEE}, {\bf 54}, 1490-1491, (1966).

\bibitem{Da} L. D. Davisson,
``Comments on `Sequence time coding for data compression',''
{\em Proc. IEEE}, {\bf 54}, 2010, (1966).

\bibitem{JH}
R. Jozsa, M. Horodecki, P. Horodecki and R. Horodecki,
``Universal Quantum Information Compression,''
{\em Phys. Rev. Lett.}, {\bf 81}, 1714 (1998);
quant-ph/9805017 (1998).
\bibitem{Expo-s}
M. Hayashi, 
``Exponents of quantum fixed-length pure state source coding,'' 
{\em Phys. Rev. A}, {\bf 66}, 032321 (2002).

\bibitem{HayaMa}
M. Hayashi and K. Matsumoto,
``Quantum universal variable-length source coding,''
{\em Phys. Rev. A} {\bf 66}, 022311 (2002).

\bibitem{Verdu-Han}
S.\ Verd\'{u} and T.S.\ Han,
``A general formula for channel capacity,''
{\em IEEE Trans. Inform. Theory}, {\bf 40}, 
1147--1157 (1994).

\bibitem{Haya97}
M. Hayashi:
``Asymptotics of quantum relative entropy 
from a representation theoretical viewpoint,''
{\em J. Phys. A: Math. and Gen.}, {\bf 34}, 3413-3419 (2001).

\bibitem{KeylW}
M. Keyl and R. F. Werner,
``Estimating the spectrum of a density operator,''
{\em Phys. Rev. A}, {\bf 64}, 052311 (2001).

\bibitem{H2001}
M. Hayashi,
``Optimal sequence of POVMs in the sense of Stein's lemma in quantum
hypothesis,''
{\em J. Phys. A: Math. and Gen.}, {\bf 35}, 10759-10773 (2002).

\bibitem{Sanov}
I. Bjelakovi\'{c}, J.-D. Deuschel, T. Kruger, R. Seiler, R. Siegmund-Schultze, and A. Szko{\l}a,
``A Quantum Version of Sanov's Theorem,''
{\em Comm. Math. Phys.}, {\bf 260}, 659-671 (2005).

\bibitem{MaHa07}
K. Matsumoto, and M. Hayashi, ``Universal distortion-free entanglement concentration,'' {\em Physical Review A}, {\bf 75}, 062338 (2007).

\bibitem{Christandl}
M. Christandl,
``The Structure of Bipartite Quantum States - Insights from Group Theory and Cryptography,''
PhD thesis, February 2006, University of Cambridge,
quant-ph/0604183.

\bibitem{KS}
C. Krattenthaler and P. Slater
``Asymptotic Redundancies for Universal Quantum Coding,''
{\em IEEE Trans. Inform. Theory}, {\bf 46}, 
801-819 (2000).

\bibitem{CB}
B. S. Clarke and A. R. Barron,
``Information-theoretic asymptotics of Bayes methods,''
{\em IEEE Trans. Inform. Theory}, {\bf 36}, 453--471 (1990).

\bibitem{prep}
M. Hayashi, 
Universal approximation of multi-copy states and universal quantum lossless data compression, arXiv:0806.1091.

\bibitem{Expo-c}
M. Hayashi,
``Error exponent in asymmetric quantum hypothesis testing and its application to classical-quantum channel coding,''
{\em Phys. Rev. A}, {\bf 76}, 062301 (2007).


\end{thebibliography}
\end{document}